\documentclass[12pt, a4paper]{article}
\usepackage{graphicx} 
\usepackage{parskip} 
\usepackage{braket} 
\usepackage{slashed}
\usepackage{wrapfig}
\usepackage[backend=biber,style=nature,doi=true,url=true]{biblatex}
\usepackage{xurl}
\usepackage[hidelinks]{hyperref}

\usepackage{setspace}
\usepackage[margin=1.15in]{geometry}

\usepackage{amsmath,amssymb,amsthm,mathtools}
\usepackage{enumitem}
\usepackage{stmaryrd} 
\usepackage{microtype}

\newtheorem{theorem}{Theorem}
\newtheorem{proposition}{Proposition}
\newtheorem{lemma}{Lemma}
\newtheorem{corollary}{Corollary}
\theoremstyle{definition}
\newtheorem{definition}{Definition}
\theoremstyle{remark}

\theoremstyle{definition}
\newtheorem{assumption}{Assumption}

\newcommand{\Pow}{\mathcal P}
\newcommand{\True}{\mathsf{True}}

\newcommand{\F}{\mathcal F}

\title{The Logical Structure of Physical Laws: A Fixed Point Reconstruction}
\author{
    Eren Volkan Küçük\footnote{Department of Physics and Astronomy, Universität Heidelberg, Germany, eren.kucuk@stud.uni-heidelberg.de}
    }
\date{}

\begin{document}
\maketitle

\begin{abstract}
We formalise the self-referential idea of physical lawhood (“self-subsumption”) as the requirement that a theory contains exactly those candidate laws that are admissible relative to the theory itself. We show that a naive extensional formulation collapses into a Russell-style type confusion (and, if forced, a self-membership pathology). The repair is to distinguish the type of law-candidates from the type of law-packages and to encode admissibility as an operator $\F$ on a lattice of packages; self-subsumption then becomes the fixed-point equation $S=\F(S)$. Under standard assumptions (packages forming a complete lattice and $\F$ monotone), Tarski’s fixed point theorem yields a canonical least fixed point $\mu F$, interpreted as a minimal stable theory under the chosen admissibility constraints. We construct broad classes of such monotone admissibility operators from invariance principles via Galois connections, and we illustrate the schema with toy instantiations inspired by quantum electrodynamics and general relativity that capture symmetry and locality constraints. We do not propose a procedure that derives the laws of nature from first principles; rather, we provide a general fixed-point architecture for reconstructing stable theory-packages from explicit admissibility criteria.
\end{abstract}


\section{Introduction}
\subsection*{Background and positioning}
The contemporary debate on laws of nature spans two broad families of views. On the one hand, Humean and “best system” approaches treat laws as those generalisations (or axioms) that best balance virtues such as simplicity and strength within an overall deductive system \cite{Carroll2020,Lewis1973,Lewis1983}. On the other hand, anti-Humean “governing” accounts aim to ground lawhood in features of the world over and above mere regularities, for instance by appealing to relations among universals or other truth-makers for nomological statements \cite{Dretske1977,Tooley1977,Armstrong1983}. A further line of critique stresses that many celebrated “laws” in physics function via idealisation, ceteris paribus clauses, or model-based approximations rather than strict exceptionless truths \cite{Cartwright1983}.\footnote{For a structured overview of these positions and the role of the best-system tradition in particular, see \cite{Carroll2020}. For an account emphasising the modal and counterfactual robustness of laws and their connection to scientific explanation, see \cite{Lange2009}.
}

The present work is intended to be largely orthogonal to these metaphysical commitments. Instead of proposing a new analysis of what laws “are”, we provide a general formal template for law-selection. We represent candidate theories as sets of law-candidates and define an admissibility operator that maps any provisional theory to the set of candidates admissible relative to it. This isolates, in a single mathematical object, the various constraints that are often invoked in practice, including symmetry and invariance considerations \cite{Earman2004,BradingCastellani2003}, while allowing the admissibility notion to depend on the current theory.

Technically, our construction relies on standard fixed-point machinery for monotone operators on complete lattices, where the least fixed point serves as a canonical “minimal stable” theory \cite{Tarski1955}. This perspective is closely related to the use of closure operators induced by adjoint pairs (Galois connections) and to fixed-point semantics familiar from program analysis \cite{CousotCousot1977}. It is also reminiscent, at the level of method rather than subject matter, of fixed-point constructions employed to manage semantic self-reference \cite{Kripke1975,Tarski1956}. Our focus, however,
is not the semantics of truth predicates; it is the reconstruction of a stable, constraint-closed theory from an explicit admissibility criterion.

\subsection*{From invariance to self-subsumption}
Another recurring thought in both metaphysics and physics is that lawhood is tied to a kind of stability: Laws are not merely true regularities, but statements (or structures) that remain robust under appropriate changes of description, coordinate choice, or background assumptions. In philosophy, this idea is often expressed by appealing to invariance. What is objective, and what counts as a genuine explanatory pattern, should be invariant under an admissible class of transformations \cite{Nozick2001,Nozick1981}. In physics the same theme appears concretely as symmetry and constraint principles. A mature theory is not just a list of dynamical equations; it is a package that is closed under the consequences of its constitutive requirements, such as symmetry, locality, unitarity, renormalisability, covariance, etc.

This paper studies a particularly strong way of packaging that stability idea, which is called \textit{self-subsumption}, that is, the proposal that the selected laws are exactly those candidates that satisfy an admissibility criterion whose evaluation itself depends on the totality of selected laws. Schematically, one tries to say
\[
\text{"}\psi\in P \iff \text{$\psi$ is admissible relative to $P$"}.
\]
Read naively and extensionally, this easily collapses into a Russell-style type confusion. The criterion is applied to an object of the wrong kind, and if one forces well-formedness by collapsing levels, the resulting statement degenerates into a toxic self-membership condition of the form $P\in P$ \cite{Russell1908}. Our first aim is to make this failure precise and to isolate what, in the extensional encoding, turns an intended explanatory constraint into either ill-typed syntax or a trivial self-endorsement.

Our second aim is to give a repair that preserves the self-subsumption idea while preventing the type collapse. The central move is to keep distinct the type of \emph{laws} (elements of a candidate domain) and the type of \emph{law-packages} (sets, subspaces, or other structured collections of laws). Once this distinction is enforced, the self-subsumption biconditional becomes a fixed point equation for an induced admissibility operator. In the basic setting where packages form a complete lattice and the admissibility operator is monotone, Tarski's theorem guarantees fixed points and provides canonical extremal solutions, in particular a least fixed point $\mu\F$ \cite{Tarski1955}. This brings self-reference into a familiar and mathematically controlled form, analogous in spirit (though not identical in purpose) to fixed point treatments of semantic self-reference \cite{Kripke1975}.

Our third aim is to explain how this fixed-point normal form functions as a useful template in physics. The formalism is not proposed as a machine that derives the laws of nature from first principles. Rather, it isolates a structural pattern already implicit in practice: Once a candidate domain $\Sigma$ is fixed and one specifies admissibility constraints, one obtains a completion map $\F$ whose fixed points are exactly the closed, admissible theory-packages. Importantly, such completion maps need not be inflationary. A "seed" $S$ may serve only to extract constraints that determine an admissible completion $\F(S)$, without requiring $S\subseteq \F(S)$. This feature matches common physics methodology, where a partial commitment can select a canonical completion under additional principles.

The main results can be summarised as follows:
\begin{enumerate}[label=(\roman*)]
\item We show that an extensional set-theoretic formulation of self-subsumption either becomes
ill-typed or collapses into self-membership once levels are identified.
\item We give a typed reconstruction in which self-subsumption is equivalent to a fixed point
equation $P=\F(P)$ on a lattice of packages, with canonical minimal choice $P=\mu\F$ under
monotonicity. \item We provide an invariance-based construction of monotone admissibility operators via a Galois correspondence between packages and symmetry groups, and illustrate the framework with toy instantiations for QED and GR. 
\end{enumerate}
We should also note that we do not claim uniqueness or correctness of $\Sigma$ or $\F$; those encode substantive physical choices. Our contribution is the fixed-point architecture and its invariance-based construction of admissibility operators. 

In that light, the plan of the paper is as follows: Section~2 diagnoses the failure of the extensional formulation. Section~3 develops the typed fixed point reconstruction and establishes the fixed point normal form. Section~4 explains how to instantiate the schema in physics via completion operators encoding constitutive constraints. Section~5 presents the QED and GR examples. We conclude by summarising what the fixed point viewpoint clarifies, and what it does not attempt to settle about the metaphysics of laws.

\section{Failure of Extensional Set-Theoretic Formulation}

Our first aim is to show that an extensional set-theoretical formulation\footnote{A property $C$, represented extensionally in set theory, is understood as the set of all objects that satisfy it; that is
\[
\text{ext}(C) := \{\psi \in \Lambda : C(\psi)\}.
\]}  of self-subsumption leads to inconsistencies, or is just unfruitful, in the sense of being a tautology. More concretely, it shows that a naive extensional set-theoretic encoding either (i) is ill-typed unless one collapses levels, or (ii) collapses into $P\in P$ once levels are collapsed.

Let $\Lambda$ be a set of candidate laws.\footnote{
We do not assume a metaphysically privileged, objective "set of all possible laws". Since this text is concerned with the formal structure of a physics theory, what concerns us is the existence of such a set, be it in a specified domain of discourse about nature or in a more general domain of metaphysics, rather than its actual law content. A standard way to construct such a set is by fixing a countable formal language $L$ and letting $\mathrm{Sent}(L)$ be its set of sentences. Since sentences are finite strings over a countable alphabet, $\mathrm{Sent}(L)$ is countable. Thereafter, it is ordinary to fix a Gödel coding $g:\mathrm{Sent}(L)\to\mathbb N$ and set
\[
\Lambda := g(\mathrm{Sent}(L))\subseteq \mathbb N .
\]
A law-candidate is then an element $\psi\in\Lambda$.

Otherwise, in a more directly physics-faced setting, one can devise $\Lambda$ as a set of (possible) operators and/or states of a system, which satisfy certain constraints, such as locality, unitarity, gauge equivalence, etc.  

Afterwards, the semantic meaning can be fixed by a class of intended models $\mathsf{Mod}$ and a mapping called a satisfaction relation. Each $\psi\in\Lambda$ has meaning via this satisfaction relation. Writing $\varphi=g^{-1}(\psi)$, define
\[
\llbracket \psi\rrbracket := \{ M\in \mathsf{Mod} : M\models \varphi\}\subseteq \mathsf{Mod}.
\]
Thus, candidate laws are meaningful insofar as they constrain the intended models.} Represent the characteristic $C$ extensionally by its extension
\[
C^{\mathrm{ext}}\subseteq \Lambda,
\qquad \psi\in C^{\mathrm{ext}} \iff C(\psi).
\]
Define the "selected" laws
\[
P \;:=\; \{\psi\in\Lambda : C(\psi)\}.
\]
\begin{lemma}\label{lem:collapse}
Under the extensional reading above, $P=C^{\mathrm{ext}}$.
\end{lemma}
\begin{proof}
For all $\psi\in\Lambda$,
\(
\psi\in P \iff C(\psi) \iff \psi\in C^{\mathrm{ext}}.
\)
Hence $P=C^{\mathrm{ext}}$ by extensionality.
\end{proof}
The informal move "assume $C(P)$" requires that $P$ is a legitimate argument of $C(\cdot)$. In the present encoding, $C(\cdot)$ is a predicate on elements of $\Lambda$, so $C(P)$ is meaningful only if
\[
P \in \Lambda.
\]
But $P$ was constructed as a subset $P\subseteq \Lambda$, i.e. a law-set (a set of laws, possibly defined under certain criteria) rather than a law. Treating $P$ simultaneously as an element of $\Lambda$ is a type collapse \cite{Russell1908}.

One can force this by choosing a universe $\Lambda$ that contains sets among its "laws". In that case, the extensional collapse becomes toxic:

\begin{proposition}
Assume the extensional setup above, and also assume the type-collapsing condition $P\in\Lambda$ so that $C(P)$ is well-formed. Then
\[
C(P)\quad\Longrightarrow\quad P\in P.
\]
\end{proposition}
\begin{proof}
By definition, $C(P)$ means $P\in C^{\mathrm{ext}}$. By Lemma~\ref{lem:collapse}, $C^{\mathrm{ext}}=P$, hence $P\in P$.
\end{proof}

\begin{corollary}
In ZFC with Foundation, the conjunction of the extensional reconstruction and $C(P)$, made well-formed by $P\in\Lambda$, is inconsistent.
\end{corollary}
In non-well-founded set theories, $P\in P$ can be consistent \cite{Aczel1988}. However, in the extensional encoding, the intended explanation
\[
\text{"$P$ is true because $C(P)$"}
\]
reduces to
\[
\text{"$P$ is true because $P\in P$"}.
\]
Since this provides no independent ground beyond $P$'s self-membership, it amounts to a trivial self-endorsement, and thus makes explicit the unfruitfulness of the extensional formulation.

\section{A Typed Fixed Point Reconstruction}

The central repair to this problem is to keep types distinct:
\[
\text{(laws)}\;\;\psi\in\Lambda
\qquad\text{vs.}\qquad
\text{(law-sets)}\;\;S\subseteq\Lambda.
\]
\begin{definition}
Let $\Lambda$ be a set of Gödel codes of candidate laws (sentences). Let $\Psi:=\Pow(\Lambda)$ be the complete lattice of law-sets ordered by $\subseteq$.\footnote{The powerset $(\mathcal P(\Lambda),\subseteq)$ is a complete lattice. For any family $\{S_i\}_{i\in I}\subseteq\mathcal P(\Lambda)$, arbitrary joins are unions and arbitrary meets are intersections \cite{davey2002}: 
\[
\bigvee_{i\in I} S_i = \bigcup_{i\in I} S_i,
\qquad
\bigwedge_{i\in I} S_i = \bigcap_{i\in I} S_i.
\]
}
\end{definition}

\begin{definition}
Let $C:\Psi\times\Lambda\to\{\top, \bot \}$ be a Boolean-valued predicate depending on a background law-set $S$ and a candidate law $\psi$.\footnote{In the metalanguage, $C(S,\psi)=\top$ is read as "relative to the background law-set $S$, the candidate law $\psi\in\Lambda$ satisfies the criterion $C$" (and $C(S,\psi)=\bot$ as its negation). } Define
\[
\F:\Psi\to\Psi,\qquad
\F(S):=\{\psi\in\Lambda : C(S,\psi)\}.
\]
\end{definition}

\begin{assumption}
\label{ass:mono}
Assume $\F$ is monotone\footnote{Many tempting criteria are not monotone in $S$. For instance, if $C(S,\psi)$ means "$S\cup\{\psi\}$ is consistent", then enlarging $S$ can break consistency, so monotonicity may fail. By contrast, if one takes
\[
C(S,\psi)\; \equiv\; (S \vdash \psi),
\]
then $C$ is monotone in its first argument, that is, $S\subseteq T$ and $S\vdash \psi$, then also $T\vdash \psi$ for any standard, structural consequence relation. For this reason, and also the intricacies coming up with Gödel's incompleteness theorems, we are not going to consider consistency as an internal criterion; that is, a criterion to be assessed \emph{in} the formal language itself, but keep it as a metalanguage criterion.
}:
\[
S\subseteq T \;\Rightarrow\; \F(S)\subseteq \F(T).
\]
\end{assumption}
Since $(\Psi,\subseteq)$ is a complete lattice and $\F$ is monotone, fixed points exist:
\begin{theorem}[Tarski's fixed point theorem \cite{Tarski1955} in specialised form]
Let $(\Psi,\subseteq)$ be a complete lattice and $\F:\Psi\to\Psi$ monotone. Then the set of fixed points of $\F$ is a complete lattice; in particular, there exists a least fixed point, denoted $\mu\F$.
\end{theorem}
By Tarski's fixed point theorem, the set of all fixed points of a monotone operator $\F:\Psi \to \Psi$ forms a complete lattice. The least fixed point is given by
\[
    \mu \F \;=\; \bigcap \{\, S \in \Psi \;:\; \F(S)=S \,\}.
\]
Thus
\[
    \F(\mu \F) = \mu \F,
    \qquad
    \forall S \in \Psi \;\; \big(\F(S)=S \;\Rightarrow\; \mu \F \subseteq S\big).
\]
\begin{lemma}
\label{lem:bottom-lfp}
Let $(\Psi,\subseteq)$ be a complete lattice with least element $\bot_\Psi$, and let $\F:\Psi\to\Psi$ be monotone. If $\F(\bot_\Psi)$ is a fixed point, i.e.\ $\F(\F(\bot_\Psi))=\F(\bot_\Psi)$, then
\[
\mu \F \;=\; \F(\bot_\Psi).
\]
\end{lemma}
\begin{proof}
Let $X$ be any fixed point, so $\F(X)=X$. Since $\bot_\Psi\subseteq X$, monotonicity gives $\F(\bot_\Psi)\subseteq \F(X)=X$. Thus $\F(\bot_\Psi)$ is contained in every fixed point. If $\F(\bot_\Psi)$ is itself fixed, it is the least fixed point.
\end{proof}
\begin{definition}
Define the deepest law-set by
\[
P \;:=\; \mu\F.
\]
\end{definition}
We will use only the defining fixed point property:
\[
\F(P)=P.
\]
Now, to define the membership, we set the following proposition:
\begin{proposition}
For all $\psi\in\Lambda$,
\[
\psi\in P \iff \psi\in\F(P) \iff C(P,\psi).
\]
\end{proposition}
Thus $P$ "explains" its members in the following limited but precise sense:
\[
\psi\in P \;\Rightarrow\; C(P,\psi).
\]
For a law to "contain its own truthmaker" in the same formulation as the law itself, let $L$ be a fixed formal language and let $\Lambda\subseteq\mathbb N$ be a set of Gödel codes for $L$-sentences. Let $p\in\Lambda$ be the code of a sentence that expresses the meta-principle internally by quantifying over codes. Schematically, we may write
\[
p \;\equiv\; \forall n\in\Lambda\;\Big(C(P,n)\ \Rightarrow\ \mathsf{Obtain}(n)\Big),
\]
where $\mathsf{Obtain}(n)$ is read as "the law coded by $n$ obtains" (i.e.\ is among the selected/actualised laws).\footnote{
We use $\mathsf{Obtain}$ rather than an object-language truth predicate $\True$ to avoid importing a separate theory of truth. A total $\True$ for the same language is delicate in strong settings (Tarski-style undefinability), while Kripke-style approaches treat $\True$ as partial via a fixed point semantics. One could nevertheless replace $\mathsf{Obtain}(n)$ by a suitably interpreted $\True(n)$ (e.g.\ metalinguistically as satisfaction in the intended model(s), or internally via a Kripke fixed point). This substitution does not affect the fixed point structure developed in the paper, since our main results concern closure of law-sets under the admissibility operator $\F(S)=\{n\in\Lambda: C(S,n)\}$ and the canonical selection $P=\mu \F$, rather than properties specific to truth predicates. \cite{Tarski1955,Kripke1975}
}
\begin{proposition}
If $C(P,p)$ holds, then $p\in P$.
\end{proposition}
\begin{proof}
If $C(P,p)$ holds, then $p\in\F(P)$. Since $P$ is a fixed point, $\F(P)=P$, hence $p\in P$.
\end{proof}
The "principle includes itself" statement is now the typed membership $p\in P$ (a sentence in a set of sentences), not $P\in P$ (a set a member of itself). Thus, the Foundation problem does not arise.

Our construction used $\Psi=\Pow(\Lambda)$, as a set of Gödel codes of candidate law-sentences, but the underlying pattern is more general. Let $(T,\le)$ be a partially ordered set of candidate theories (or "law-sets"). Intuitively, $S\le T$ means "$T$ has at least as much content/commitment/information as $S$".

To model a self-subsuming principle, we need a way to say when an item (law, constraint, statement, datum) is admissible relative to a candidate theory.
\begin{definition}
Let $U$ be a universe of candidate items (laws, constraints, axioms, sentences, parameters, etc.). An admissibility predicate is any relation
\[
D \;\subseteq\; T\times U,
\qquad
D(S,u)\ \text{read: "$u$ is admissible/endorsed given $S$"}.
\]
\end{definition}
When $T$ is a space of law-sets (e.g.\ $T=\Pow(\Lambda)$), one typically identifies $U$ with the same kind of objects (e.g.\ $U=\Lambda$) and regards $S$ as endorsing exactly those $u$ satisfying $D(S,u)$.\footnote{In the following section, where we begin applying our schema to physics, we are going to consider operator algebras instead of logical sentences or Gödel codes of them, for it is both eaiser and more usual to untangle the algebraic representations pointing to the same phenomena than doing the same for the sentence/code representation in a Principia Mathematica type of formal language} This motivates the associated operator:
\begin{definition}
Assume $T$ is a collection of subsets of $U$.
Define an operator
\[
\F_D:T\to T,
\qquad
\F_D(S):=\{u\in U: D(S,u)\}.
\]
\end{definition}
\begin{proposition}
\label{prop:normal-form}
For any admissibility predicate $D$ and induced operator $\F_D$, a candidate $P\in T$ satisfies the self-subsumption biconditional
\[
\forall u\in U\;\;\big(u\in P \iff D(P,u)\big)
\]
if and only if $P$ is a fixed point of $\F_D$, i.e.\ $\F_D(P)=P$.
\end{proposition}
\begin{proof}
By definition, $u\in \F_D(P)\iff D(P,u)$, hence
\[
u\in P \iff D(P,u)
\quad\Longleftrightarrow\quad
u\in P \iff u\in \F_D(P),
\]
which is equivalent to $P=\F_D(P)$ by extensionality of sets.
\end{proof}
Proposition~\ref{prop:normal-form} says any attempt to characterise a fundamental law-set $P$ by a membership rule that itself refers to $P$ is exactly a fixed point equation. Our $C(P,\psi)$ is one instance with $D=C$ and $U=\Lambda$. No monotonicity or lattice assumptions are required for Proposition \ref{prop:normal-form} to hold.

By contrast, the existence of fixed points and the availability of a canonical choice, such as the least fixed point $\mu \F$, do require additional hypotheses. In the closure-style regime considered here, we assume $\F$ is monotone on a complete lattice, so that Tarski guarantees fixed points and yields distinguished extremal solutions. Non-monotone admissibility notions (for instance consistency or optimality constraints) may still admit fixed points but fall outside the scope of the Tarski canonicity argument. For this reason, to get the existence of fixed points systematically, one typically assumes a lattice-theoretic structure and monotonicity. At the beginning of this section, we have a complete lattice $\Psi=\Pow(\Lambda)$ and (by Assumption~\ref{ass:mono}) a monotone operator $\F$. Hence, Tarski applies and yields least and greatest fixed points; in particular $P=\mu\F$ is well-defined.

Importantly, Tarski-style existence does not require uniqueness. In fact, the general conclusion is the opposite: Fixed points can form a rich lattice. If our aim is to analyse where such assumptions lead, this is a feature:
\begin{enumerate}[label = \textit{\roman*}.]
\item different fixed points correspond to different admissible global law-sets satisfying the same meta-criterion,\footnote{One can conceive a duality \textit{à la} De Haro and Butterfield \cite{deharo2025duality, deharo2025geometricviewtheories}, where any such distinct fixed point characterisation of a theory is isomorphic.}
\item the least and greatest fixed points give canonical extremal solutions,
\item studying the lattice of fixed points becomes a principled way to study the space of possible "fundamental" packages consistent with a given admissibility concept.
\end{enumerate}
So the fixed point framework is naturally a constraint-analysis engine. The genuinely structural content is:
\[
\text{(self-subsumption schema)} \quad\Rightarrow\quad \text{(fixed point equation)}.
\]
Choosing which fixed point (least, greatest, or another) is an extra selection principle. Our choice of $\mu \F$ corresponds to a minimality stance.

The equation $P=\F(P)$ is still a circular characterisation. The fixed point method does not erase circularity; it replaces it with a non-vicious circularity\footnote{
Here "non-vicious circularity" means that self-reference is implemented as a \emph{constraint/closure} condition rather than as a mere justificatory loop. That is, one specifies a monotone operator $\F$ and then takes the least fixed point $\mu \F$, so the circularity $P=\F(P)$ determines $P$ by minimal closure rather than presupposing it. This is the familiar pattern behind Kripke-style fixed point treatments of self-reference in truth theory, and it aligns with contemporary grounding discussions in which some circles need not be vicious, once a principled notion of "viciousness" is adopted. \cite{Tarski1955,Kripke1975,Bliss2013}
}
resolved by a selection principle, that is, "choose the least fixed point $P=\mu\F$".

Given $\F$, $P=\mu\F$ is canonical. But the choice of $\F$ (equivalently, of the admissibility criterion $C$) is a remaining datum. In another light, such a choice would not say much about the structural content of the theory, and would work as a choice of gauge.

\section{Physical Instantiation}
The fixed point normal form is a meta-constraint template. To apply it to physics, fix (i) a candidate domain of laws $\Sigma$ and (ii) a lattice of theory-packages $\mathcal T$ (such as $\mathcal P(\Sigma)$, or as in the Examples section below, lattice of linear subspaces when $\Sigma$ is a vector space) together with an admissibility operator $\F:\mathcal T\to\mathcal T$ encoding the chosen constraints. A package counts as admissible precisely when it is closed:
\[
P=\F(P).
\]
In physics, certain requirements function less like contingent dynamical laws and more like constitutive constraints. Examples include unitarity, symmetry principles, locality/causality constraints, renormalisability, etc. A natural reading of our framework is:
\begin{quote}
\emph{A package counts as a physical theory only if it is closed under the implications of these admissibility constraints.}
\end{quote}
Thus, $\F(S)$ returns the admissible completion determined by the constraints extracted from $S$.\footnote{Note that it does not always have to contain $S$ itself once it is used as a seed to select an admissible completion, as in the case of the QED example in the next section.} As an intuitive example, let $\mathcal T$ be a space of partial "theory data" ordered by "information content". Define $\F(S)$ to be the admissible completion operator of $S$ under a certain set of admissibility constraints, which ought to be chosen by the physicist according to the subject in question. Then, fixed points are solutions that give canonical admissible packages; that is, $S$ is admissible iff $\F(S)=S$.

In this reading, our framework organises the landscape such that the fixed points are the admissible theories, and extremal solutions ($\mu \F$, $\nu \F$) give canonical minimal structures.

Therefore, one can state a "metaprinciple" as follows:
\begin{quote}
\textbf{Meta-principle:}
Any proposal in which "the laws of nature" are characterised as exactly those statements satisfying an admissibility criterion that itself depends on the totality of laws has, in general, fixed point form $P=\F(P)$, with a canonical minimal choice, $\mu \F$.
\end{quote}

\subsection{A General Logical Architecture}
While the logical schema outlined guarantees fixed points for any monotone operator, we must identify a physically motivated operator that satisfies this property. As noted previously, simple consistency checks are often non-monotone and metatheoretically complex. In this section, we propose that the interplay between \emph{Dynamical Laws} and \emph{Symmetry Groups} provides the robust, monotonic structure required to ground self-subsuming principles.

Let us define two abstract domains of discourse for a physical theory:
\begin{definition}[The Space of Laws, $\Sigma$]
Let $\Sigma$ be a chosen candidate set of possible mathematical formulations of dynamical constraints.
\end{definition}

\begin{definition}[The space of symmetry groups, $\Gamma$]
Let $\Gamma$ be a poset of admissible symmetry groups acting on the system, ordered by inclusion. In typical applications $\Gamma$ is a family of subgroups of a fixed ambient transformation group (e.g.\ internal gauge transformations, or $\operatorname{Diff}(M)$), and we assume $\Gamma$ is closed under the joins needed below (e.g.\ it is a complete lattice, or at least admits the join of the relevant collections of subgroups).\footnote{Depending on the context, elements of $\Sigma$ could be terms in a Lagrangian, differential equations of motion, or elements of a $C^*$-algebra; and that of $\Gamma$ could be rotation groups, gauge groups, or diffeomorphism group, etc.

A further input is an equivalence notion capturing when two presentations encode the same physics, for instance, via gauge equivalence, field redefinitions, dualities. One may enforce this either by quotienting the candidate domain $\Sigma$ by the chosen equivalence, or by requiring the admissibility criterion to be invariant under it.
}
\end{definition}
Next, we posit  a binary relation $R\subseteq \Sigma\times \Gamma$, where $\phi\,R\,G$ means that the law $\phi$ respects the symmetry group $G$.\footnote{In physics terminology, this can be understood as "$\phi$ is invariant under the group $G$", with invariance understood modulo the chosen equivalence on $\Sigma$.} This relation induces two canonical mappings:
\begin{definition}[Invariant map, $\operatorname{Inv}:\Gamma\to \mathcal T$]
For $G\in\Gamma$, define
\[
\operatorname{Inv}(G)\;:=\;\{\phi\in \Sigma:\ \phi\,R\,G\},
\]
viewed as an element of the package lattice $\mathcal T$ (e.g.\ if $\mathcal T$ is a lattice of linear subspaces, then $\operatorname{Inv}(G)$ is automatically a subspace because linear combinations of invariant laws are invariant).
\end{definition}
\begin{definition}[Symmetry map, $\operatorname{Sym}:\mathcal T\to\Gamma$]
For $S\in\mathcal T$, define
\[
\operatorname{Sym}(S)\;:=\;\bigvee \{\,G\in\Gamma:\ S\subseteq \operatorname{Inv}(G)\,\},
\]
i.e.\ the greatest symmetry in $\Gamma$ compatible with $S$ (the join exists by the standing closure assumption on $\Gamma$).
\end{definition}
With these definitions, one has the Galois correspondence
\[
S\subseteq \operatorname{Inv}(G)\quad\Longleftrightarrow\quad G\subseteq \operatorname{Sym}(S),
\]
so $\operatorname{Inv}$ and $\operatorname{Sym}$ form an antitone Galois connection between packages and symmetry groups. In particular, $\operatorname{Inv}$ is antitone and $\operatorname{Sym}$ is antitone, hence their composition is monotone. The logic follows directly from the reversal of the inclusion order twice:
\[
    S_1 \subseteq S_2 \implies \operatorname{Sym}(S_2) \subseteq \operatorname{Sym}(S_1) \implies \operatorname{Inv}(\operatorname{Sym}(S_1)) \subseteq \operatorname{Inv}(\operatorname{Sym}(S_2)).
\]
We then obtain a canonical invariance-based completion operator
\[
\F(S)\;:=\;\operatorname{Inv}(\operatorname{Sym}(S)).
\]
This operator is monotone on $\mathcal T$, so by Tarski's theorem it admits fixed points. Fixed points $P=\F(P)$ are precisely packages that are symmetry-closed relative to the chosen candidate domain and admissible symmetry family.

More general admissibility operators may be built by composing $\operatorname{Sym}$ and $\operatorname{Inv}$ with additional monotone meta-maps, such as the localisation map on $\Gamma$, encoding a gauge principle in the QED example below. Such operators remain monotone but need not satisfy $S\subseteq \F(S)$, so a seed may select an admissible completion without being contained in it.

\subsection{Cases with More Structural Constraints}
It is important to emphasise that the specific Law--Symmetry duality need not be the only structural constraint on a physical theory. In practice, mature theories must satisfy multiple independent requirements, such as symmetry (gauge, Lorentz, diffeomorphism, etc.), causality, and unitarity, simultaneously. A strength of the fixed point schema is that such requirements can be composed into a single admissibility operator.

In the abstract discussion, the package lattice $\mathcal T$ need not be a powerset $\Pow(\Sigma)$. In the Examples section, for instance, $\Sigma$ is a vector space of candidate Lagrangian terms and $\mathcal T$ is taken to be the complete lattice of linear subspaces of $\Sigma$. In such cases, an operator defined naively by
\[
S\longmapsto \{\psi\in\Sigma:\ C(S,\psi)\}
\]
need not land in $\mathcal T$. To define admissibility operators uniformly on an arbitrary package lattice $\mathcal T$, we introduce a \emph{package closure} map that projects subsets of $\Sigma$ to the least package containing them.

\begin{definition}
Let $\mathcal T\subseteq \Pow(\Sigma)$ be a complete lattice of packages ordered by $\subseteq$.
Define the closure operator
\[
\langle X\rangle_{\mathcal T}\;:=\;\bigwedge\{\,P\in\mathcal T:\ X\subseteq P\,\}\in\mathcal T,
\qquad X\subseteq \Sigma.
\]
\end{definition}
If $\mathcal T=\Pow(\Sigma)$ then $\langle X\rangle_{\mathcal T}=X$. If $\Sigma$ is a vector space and $\mathcal T$ is the lattice of linear subspaces, then $\langle X\rangle_{\mathcal T}=\mathrm{span}(X)$.

Each physical principle may be encoded as a monotone operator
\[
\F_i:\mathcal T\to\mathcal T,
\]
for instance via a predicate $C_i(S,\psi)$ by setting
\[
\F_i(S)\;:=\;\Big\langle \{\psi\in\Sigma:\ C_i(S,\psi)\}\Big\rangle_{\mathcal T}.
\]

Given a family of monotone operators $\{\F_i\}_i$, define their combined operator by the meet in $\mathcal T$:
\[
\F_{\mathrm{tot}}(S)\;:=\;\bigwedge_i \F_i(S).
\]
In the examples, $\mathcal T$ is closed under intersections, so this meet is simply $\bigcap_i \F_i(S)$.

\begin{lemma}
If each $\F_i$ is monotone, then $\F_{\mathrm{tot}}$ is monotone.
\end{lemma}
\begin{proof}
If $S\subseteq S'$ then $\F_i(S)\subseteq \F_i(S')$ for all $i$.
Hence $\bigwedge_i \F_i(S)\subseteq \bigwedge_i \F_i(S')$, so $\F_{\mathrm{tot}}$ is monotone.
\end{proof}

Thus, on a complete lattice of theory-packages $\mathcal T$, Tarski's theorem applies to $\F_{\mathrm{tot}}$. In particular, $\F_{\mathrm{tot}}$ admits fixed points, and its least fixed point $\mu\F_{\mathrm{tot}}$ gives a canonical minimal package stable under all encoded principles.

This shows that the fixed point view is not tied to any single physical hypothesis, but provides a general logical architecture: Once a domain of candidates $\Sigma$ and admissibility constraints are specified, "the theory" is a structure in the lattice of theory-packages. Additional principles, known or yet to be discovered, correspond to refining $\F_{\mathrm{tot}}$, which updates the resulting stable packages accordingly. Of course, choosing or discovering such criteria is itself a philosophical problem to be discussed, which will not be attempted here.

Fixed points of the operator $P=\F(P)$ may be called closed theory-packages relative to $\Sigma$. The laws in $P$ determine a symmetry content $\operatorname{Sym}(P)$, and $\operatorname{Inv}(\operatorname{Sym}(P))$ returns the set of laws compatible with that symmetry content within the chosen candidate domain. Thus, $P$ contains no terms that break its admitted symmetries, and it is complete with respect to the chosen convention. Tarski's theorem guarantees a least fixed point, which is the minimal closed package relative to $\F$ (and to $\Sigma$).

In summary, the Galois connection provides a mathematically well-behaved admissibility operator built from invariance, avoiding the metatheoretic difficulties of treating raw "consistency checking" as the primary admissibility filter.

\section{Examples}
In the physics examples below, we show that once one fixes a natural candidate domain and encodes standard physical constraints (symmetry, locality, and a truncation criterion) as a monotone completion operator, familiar theories such as QED and GR can be presented as fixed points of that operator.\footnote{Here, for clarity, we should once again note that this approach is constructive; that is, the following is not an argument of derivation.} For more information on the subjects, the reader is kindly directed to \cite{Carroll2019, MisnerThorneWheeler1973, PeskinSchroeder2015, Weinberg1995}.

\subsection{I. Quantum Electrodynamics}
\paragraph{Operator domain.}
Fix four-dimensional Minkowski spacetime.\footnote{One can as well state "Lorentz invariance" as one of the symmetry principles explicitly.} Let $\Sigma_{\le 4}$ denote the vector space of local polynomial Lagrangian densities built from a Dirac spinor $\psi(x)$ and a vector potential $A_\mu(x)$, and finitely many derivatives, subject to:
\begin{enumerate}
    \item a power-counting bound of dimension $\le 4$ \cite{PeskinSchroeder2015};
    \item an equivalence relation $\sim$ identifying densities that differ by a total derivative and by standard "inessential" redundancies caused by integration by parts.
\end{enumerate}
We write $[\mathcal O]\in \Sigma_{\le 4}$ for the equivalence class of a density $\mathcal O$. This implements the idea that physics is insensitive to Lagrangian terms differing by boundary terms.
\paragraph{Theory packages.}
A theory package $S$ will be represented as a linear subspace $S\subseteq \Sigma_{\le 4}$, spanned by a chosen generating set of operator classes. In particular, if $[\mathcal O_1],[\mathcal O_2]\in S$, then $a[\mathcal O_1]+b[\mathcal O_2]\in S$ for all scalars $a,b$, where the linear closure is treated as background structure instead of a part of the admissibility operator.\footnote{The collection of linear subspaces of $\Sigma_{\le 4}$ ordered by inclusion forms a complete lattice: Meets are intersections and joins are linear spans of unions.}
\paragraph{Symmetry space.}
Let $\Gamma$ be a poset of symmetry groups acting on the fields, ordered by inclusion. Since we fix Minkowski spacetime and restrict attention to Lorentz-scalar densities, spacetime symmetries are treated as background constraints and omitted from $\Gamma$ (although it is trivial to incorporate them, as well as the discrete symmetries).\footnote{If the spinor field were not massive, it would also admit a transformation of the form
\[
    \psi \mapsto e^{i \theta \gamma^5} \psi,
\]
yielding an additional axial $U(1)_A$ symmetry at the classical level, so the continuous internal symmetry enlarges to $U(1)_V\times U(1)_A$ in that limit.
}
Since $\psi$ carries no additional flavour or internal index, the relevant internal symmetry transformations are constant, non-spacetime field redefinitions $\psi\mapsto U\psi$ that preserve the free Dirac density $\bar\psi(i\gamma^\mu\partial_\mu-m)\psi$. Preservation of the mass term requires $\bar\psi\psi\mapsto \bar\psi\psi$, and preservation of the kinetic term requires $\bar\psi\gamma^\mu\partial_\mu\psi\mapsto \bar\psi\gamma^\mu\partial_\mu\psi$. These two requirements imply that $U$ commutes with the Dirac bilinear structure, hence $U$ can only act as multiplication by a constant phase. Imposing unitarity restricts this to $U=e^{i\theta}$, so the continuous internal symmetry group is $U(1)_{\mathrm{glob}}$. The free theory has a global symmetry, $U(1)_{\mathrm{glob}}$, whereas with the addition of the gauge principle as one of the criteria, we also have a local one, $U(1)_{\mathrm{loc}}$. Therefore,
\[
\Gamma \;:=\;\{\,\mathbf{1},\,U(1)_{\mathrm{glob}},\,U(1)_{\mathrm{loc}}\,\},
\qquad
\mathbf{1}\subset U(1)_{\mathrm{glob}} \subset U(1)_{\mathrm{loc}}.
\]
Hence $\Gamma$ has a greatest element given by
\[
\top_\Gamma \;=\; U(1)_{\mathrm{loc}}.
\]
\paragraph{Invariant map.}
For $G\in\Gamma$, define
\[
\operatorname{Inv}(G)\;:=\;\{\, [\mathcal O]\in \Sigma_{\le 4} : [\mathcal O]\ \text{is invariant under}\ G \,\}.
\]
Here, invariance means that $[\mathcal O]$ is invariant if for one (equivalently any) representative $\mathcal{O}$, the transformed density differs from $\mathcal O$ by $\sim$.
\paragraph{Symmetry map.}
For a package $S\subseteq \Sigma_{\le 4}$, define
\[
\operatorname{Sym}(S)\;:=\;\max\{\,G\in\Gamma:\ S\subseteq \operatorname{Inv}(G)\,\}.
\]
Since $U(1)_{\mathrm{glob}}\subset U(1)_{\mathrm{loc}}$, this maximum exists.\footnote{In the examples, $\Gamma$ is a finite chain, so we identify $\operatorname{Sym}(S)$ with the maximal symmetry group in $\Gamma$, that leaves $S$ invariant.}
\paragraph{Galois correspondence.}
With the above definitions one has the standard antitone correspondence:
\[
S\subseteq \operatorname{Inv}(G)\quad\Longleftrightarrow\quad G\subseteq \operatorname{Sym}(S),
\]
so $\operatorname{Inv}$ and $\operatorname{Sym}$ form a Galois connection between operator packages and symmetry groups.

The step from global to local symmetry is not obtained automatically from $\operatorname{Sym}(S)$; it represents an additional physical principle, the gauge principle, encoded here as a map on symmetry groups.

\paragraph{Localisation operator.}
Introduce an order-preserving map, $\operatorname{Loc}:\Gamma\to\Gamma$, as
\[
\operatorname{Loc}(\mathbf{1})=\mathbf{1},\qquad
\operatorname{Loc}\big(U(1)_{\mathrm{glob}}\big)=U(1)_{\mathrm{loc}},\qquad
\operatorname{Loc}\big(U(1)_{\mathrm{loc}}\big)=U(1)_{\mathrm{loc}}.
\]
It should be noted that, in this example, $\operatorname{Loc}$ fixes the top element,
\[
\operatorname{Loc}(\top_\Gamma)=\top_\Gamma,
\]
reflecting the idea that once the symmetry is already local, localisation does nothing further. An element $\alpha(x) \in U(1)_\mathrm{loc}$ acts by
\[
\psi(x)\mapsto e^{i\alpha(x)}\psi(x),
\;
\bar\psi(x)\mapsto \bar\psi(x)e^{-i\alpha(x)},
\;
A_\mu(x)\mapsto A_\mu(x)-\frac{1}{e}\partial_\mu\alpha(x),
\]
where the coupling $e$ is fixed so that $D_\mu=\partial_\mu+ieA_\mu$ transforms covariantly.

The localisation map is not itself one of the admissibility predicates $C_i$, nor one of the induced package-operators $\F_i$, because it does not act on theory packages $S$ or on candidate laws $\psi$. Rather, $\operatorname{Loc}$ is part of the background structure used to define an admissibility principle associated with local gauge symmetry. Concretely, one may introduce a predicate
\[
C_{\mathrm{gauge}}(S,[\mathcal O])\;:=\;\Big([\mathcal O]\in \operatorname{Inv}\big(\operatorname{Loc}(\operatorname{Sym}(S))\big)\Big),
\]
which reads "relative to the background package $S$, the candidate operator $[\mathcal O]$ is admissible because it is
invariant under the localized symmetry extracted from $S$." The corresponding induced operator is then
\[
\F_{\mathrm{gauge}}(S)\;:=\;\{[\mathcal O]\in \Sigma_{\le 4}:\ C_{\mathrm{gauge}}(S,[\mathcal O])\},
\]
so $\operatorname{Loc}$ enters the schema via the definitional content of $C_{\mathrm{gauge}}$ and hence of $\F_{\mathrm{gauge}}$. In the present example, we fold this into the composite operator $\F_{\mathrm{QED}}(S)=\operatorname{Inv}(\operatorname{Loc}(\operatorname{Sym}(S)))$.
\paragraph{Total admissibility operator.}
Define the operator on packages
\[
\F_{\mathrm{QED}}(S)\;:=\;\operatorname{Inv}\!\big(\operatorname{Loc}(\operatorname{Sym}(S))\big)\ \subseteq\ \Sigma_{\le 4}.
\]
By construction, $\F_{\mathrm{QED}}$ is monotone with respect to $\subseteq$. Intuitively, this operator (i) diagnoses the symmetries of $S$, (ii) enforces their localisation, and (iii) outputs the theory of renormalisable operators compatible with this constraint.

\subsubsection*{The computation from a free matter seed}
Let the seed package be the span of the free Dirac kinetic term:
\[
P_0\;:=\;\mathrm{span}\Big\{\big[\bar\psi\, i\gamma^\mu\partial_\mu\psi\big]\Big\}\ \subseteq\ \Sigma_{\le 4}.
\]
Under the gauge transformation, one has
\[
\partial_\mu(e^{i\alpha(x)}\psi)
=
e^{i\alpha(x)}\big(\partial_\mu\psi+i(\partial_\mu\alpha)\psi\big),
\]
so the free term $\bar\psi i\gamma^\mu\partial_\mu\psi$ acquires an extra contribution proportional to $\partial_\mu\alpha$. Hence it is invariant for $\alpha=\text{const}$ but not for general $\alpha(x)$, and therefore
\[
\operatorname{Sym}(P_0)\;=\;U(1)_{\mathrm{glob}}.
\]
To incorporate gauge principle, apply $\operatorname{Loc}$:
\[
\operatorname{Loc}(\operatorname{Sym}(P_0))\;=\;\operatorname{Loc}\big(U(1)_{\mathrm{glob}}\big)\;=\;U(1)_{\mathrm{loc}}.
\]
Now, to compute $\operatorname{Inv}(U(1)_{\mathrm{loc}})\cap \Sigma_{\le 4}$, we list the renormalisable invariant operator classes:
\begin{enumerate}
\item \textbf{Fermion mass.}
The class $\big[\bar\psi\psi\big]$ is invariant since $e^{-i\alpha}e^{i\alpha}=1$:
\[
\bar\psi\psi \mapsto \bar\psi e^{-i\alpha} e^{i\alpha}\psi=\bar\psi\psi.
\]
\item \textbf{Photon mass}
A Proca term $\big[A_\mu A^\mu\big]$ is not invariant under $A_\mu\mapsto A_\mu-\frac{1}{e}\partial_\mu\alpha$, so $\big[A_\mu A^\mu\big]\notin \operatorname{Inv}(U(1)_{\mathrm{loc}})$. Therefore, the photon mass is excluded.
\item \textbf{Minimal coupling.}
Define the covariant derivative
\[
D_\mu := \partial_\mu + ieA_\mu.
\]
Using the transformation properties of the field and gauge field,
\[
D_\mu(e^{i\alpha}\psi)
=
(\partial_\mu+ieA_\mu)\,e^{i\alpha}\psi
=
e^{i\alpha}\big(\partial_\mu+ie(A_\mu-\tfrac{1}{e}\partial_\mu\alpha)\big)\psi
=
e^{i\alpha}D_\mu\psi,
\]
so $D_\mu\psi$ transforms like $\psi$.
Hence the class $\big[\bar\psi\, i\gamma^\mu D_\mu\psi\big]$ is invariant.
Expanding it, we have the familiar interaction term,
\[
\bar\psi\, i\gamma^\mu D_\mu\psi
=
\bar\psi\, i\gamma^\mu \partial_\mu\psi
\;-\;
e\,\bar\psi\gamma^\mu A_\mu\psi.
\]
\item \textbf{Gauge-field kinetic term.}
The field strength
\[
F_{\mu\nu}:=\partial_\mu A_\nu-\partial_\nu A_\mu
\]
is invariant because the shift piece cancels by commutativity of partial derivatives as
\[
\partial_\mu(A_\nu-\tfrac{1}{e}\partial_\nu\alpha)-\partial_\nu(A_\mu-\tfrac{1}{e}\partial_\mu\alpha)
=
\partial_\mu A_\nu-\partial_\nu A_\mu.
\]
Therefore $\big[-\tfrac14 F_{\mu\nu}F^{\mu\nu}\big]\in \operatorname{Inv}(U(1)_{\mathrm{loc}})$.
\item \textbf{No further renormalisable invariants up to equivalence.}
Dimension-$5$ operators are excluded by the $\le 4$ truncation.
A term $F_{\mu\nu}\tilde F^{\mu\nu}$ is a total derivative in the abelian case and is removed by the equivalence relation. Gauge-fixing terms are not in $\operatorname{Inv}(U(1)_{\mathrm{loc}})$; they enter only at the quantisation stage.
\end{enumerate}
\paragraph{The package.}
The image of the seed under $\F_{\mathrm{QED}}$ is
\[
\F_{\mathrm{QED}}(P_0)=\operatorname{Inv}(U(1)_{\mathrm{loc}})\subseteq \Sigma_{\le 4},
\]
whose minimal generating set may be taken as
\[
P_{\mathrm{QED}}
=
\mathrm{span}\Big\{
\big[\bar\psi\, i\gamma^\mu D_\mu\psi\big],\;
\big[\bar\psi\psi\big],\;
\big[-\tfrac14 F_{\mu\nu}F^{\mu\nu}\big]
\Big\}.
\]
In Lagrangian form, this is exactly QED:
\[
\mathcal L_{\mathrm{QED}}
=
\bar\psi(i\gamma^\mu D_\mu-m)\psi
-\frac14 F_{\mu\nu}F^{\mu\nu}.
\]
\paragraph{Fixed point property.}
The package $P_{\mathrm{QED}}$ is stable under the admissibility operator:
\[
\F_{\mathrm{QED}}(P_{\mathrm{QED}})=P_{\mathrm{QED}},
\]
since $\operatorname{Sym}(P_{\mathrm{QED}})=U(1)_{\mathrm{loc}}$, and $P_{\mathrm{QED}}$ contains precisely
the renormalisable invariants under that group within $\Sigma_{\le 4}$.
\begin{proposition}
Since $\Gamma$ has a greatest element, and $\operatorname{Loc}(U(1)_{\mathrm{loc}})=U(1)_{\mathrm{loc}}$, the package
\[
P_{\mathrm{QED}}=\operatorname{Inv}(U(1)_{\mathrm{loc}})\subseteq \Sigma_{\le 4}
\]
is the least fixed point of $\F_{\mathrm{QED}}(S)=\operatorname{Inv}(\operatorname{Loc}(\operatorname{Sym}(S)))$.
\end{proposition}
\begin{proof}
Since $\bot_\Sigma=\{0\}$ imposes no nontrivial invariance constraints, $\operatorname{Sym}(\bot_\Sigma)$ is the greatest element of $\Gamma$,
hence $\operatorname{Sym}(\bot_\Sigma)= \top_\Gamma = U(1)_{\mathrm{loc}}$. Using $\operatorname{Loc}(U(1)_{\mathrm{loc}})=U(1)_{\mathrm{loc}}$,
\[
\F_{\mathrm{QED}}(\bot_\Sigma)
=
\operatorname{Inv}(\operatorname{Loc}(\operatorname{Sym}(\bot_\Sigma)))
=
\operatorname{Inv}(U(1)_{\mathrm{loc}})
=
P_{\mathrm{QED}}.
\]
Moreover, within the chosen $\Gamma$ one has $\operatorname{Sym}(P_{\mathrm{QED}})=U(1)_{\mathrm{loc}}$, hence
\[
\F_{\mathrm{QED}}(P_{\mathrm{QED}})
=
\operatorname{Inv}(\operatorname{Loc}(\operatorname{Sym}(P_{\mathrm{QED}})))
=
\operatorname{Inv}(U(1)_{\mathrm{loc}})
=
P_{\mathrm{QED}},
\]
so $\F_{\mathrm{QED}}(\bot_\Sigma)$ is a fixed point. By Lemma~\ref{lem:bottom-lfp}, it is the least fixed point.
\end{proof}
The "matter seed" $P_0=\mathrm{span}\{[\bar\psi i\gamma^\mu\partial_\mu\psi]\}$ is not itself a gauge-invariant operator class, so it need not satisfy $P_0\subseteq P_{\mathrm{QED}}$. What the construction shows is rather that applying the completion operator $\F_{\mathrm{QED}}$ to $P_0$ produces the least fixed point in one step:
\[
\F_{\mathrm{QED}}(P_0)=\operatorname{Inv}(U(1)_{\mathrm{loc}})=P_{\mathrm{QED}}=\mu \F_{\mathrm{QED}}.
\]
Thus the example exhibits QED as the canonical least fixed point selected by the monotone operator $\F_{\mathrm{QED}}$ on the package lattice.

Thus, it is shown that QED, in this form, can be presented as a least fixed point schema under the composed metalanguage assumptions of symmetry, locality in the sense of gauge principle, and closedness under local invariants within the renormalisable operator domain.

\subsection{II. General Relativity}
The QED example showed how a local internal symmetry selects interaction terms in a renormalisable operator domain under the chosen admissibility operator. General Relativity provides an analogous illustration. Once one fixes "metric geometry" as the kinematical arena and imposes diffeomorphism invariance together with locality and a derivative-order restriction, the Einstein--Hilbert action can be presented as a minimal stable package.
\paragraph{Geometric background.}
Let $M$ be a smooth, connected, oriented $4$-manifold. The dynamical field is a Lorentzian metric $g_{\mu\nu}$ on $M$. Diffeomorphisms $\phi\in \operatorname{Diff}(M)$ act on $g$ by pullback:
\[
g \longmapsto \phi^\ast g.
\]
\paragraph{Candidate domain.}
Let $\Sigma_{\le 2}$ be the vector space of local scalar densities built from $g_{\mu\nu}$ and finitely many derivatives, restricted to at most two derivatives of the metric,\footnote{Restricting the equations of motion to second derivative order is required to avoid the appearance of negative-energy ghost modes, called Ostrogradsky instability. In four dimensions, Lovelock's theorem \cite{Lovelock1971} guarantees that the only diffeomorphism-invariant actions satisfying this stability criterion are the Einstein-Hilbert and cosmological terms.} modulo the standard equivalence relation $\sim$ identifying densities that differ by a total derivative,
\[
\mathcal L \sim \mathcal L + \nabla_\mu V^\mu.
\]
We write $[\mathcal L]\in \Sigma_{\le 2}$ for an equivalence class. The corresponding action functional is $S[g]=\int_M d^4x\, [\mathcal L]$ under appropriate assumptions so that boundary terms can be ignored.
\paragraph{Theory packages.}
A theory package is represented by a linear subspace $S\subseteq \Sigma_{\le 2}$ spanned by a chosen set of density-classes. Linearity is again treated as background structure, so closure under linear combinations does not need to be enforced by the admissibility operator.

\paragraph{Symmetry space.}
For the GR example, the maximal group acting on the metric field is the group of diffeomorphisms of $M$, $\operatorname{Diff}(M)$. Therefore, let $\Gamma$ be
\[
    \Gamma := \{\mathbf{1}, \operatorname{Diff}(M) \}.
\]
So the "largest" symmetry allowed by the chosen framework is diffeomorphism invariance.
\paragraph{Invariant map.}
Define
\[
\operatorname{Inv}(G)\;:=\;\{\, [\mathcal L]\in\Sigma_{\le 2}\ :\ [\mathcal L]\ \text{is invariant under}\ G\,\},
\]
where invariance means $\mathcal L(g)$ and $\mathcal L(\phi^\ast g)$ differ by a total derivative, and hence lie in the same equivalence class for all $\phi\in G$.

\paragraph{Symmetry map.}
For a package $S\subseteq\Sigma_{\le 2}$ define
\[
\operatorname{Sym}(S)\;:=\;\max\{\,G\in\Gamma:\ S\subseteq \operatorname{Inv}(G)\,\}.
\]
Since $\Gamma=\{\mathbf 1,\operatorname{Diff}(M)\}$ is a finite chain, this maximum exists.

With these definitions one has the standard antitone correspondence
\[
S\subseteq \operatorname{Inv}(G)\quad\Longleftrightarrow\quad G\subseteq \operatorname{Sym}(S),
\]
so $\operatorname{Inv}$ and $\operatorname{Sym}$ form a Galois connection between operator packages and symmetry groups.
\paragraph{Admissibility operator.}
Define the GR admissibility operator as
\[
\F_{\mathrm{GR}}(S)\;:=\;\operatorname{Inv}\!\big(\operatorname{Sym}(S)\big)\ \subseteq\ \Sigma_{\le 2}.
\]
This operator is monotone with respect to $\subseteq$ by general properties of the Galois correspondence.
\paragraph{Seed.}
The "seed" is the kinematical commitment to a metric field, together with no further dynamical density-terms:
\[
P_0 := \{0\}\ \subseteq\ \Sigma_{\le 2}.
\]
Because $P_0$ contains no nontrivial densities, every diffeomorphism vacuously preserves all elements of $P_0$, hence
\[
\operatorname{Sym}(P_0)=\operatorname{Diff}(M).
\]
Therefore
\[
\F_{\mathrm{GR}}(P_0)=\operatorname{Inv}(\operatorname{Diff}(M)),
\]
so computing the image reduces to classifying the diffeomorphism-invariant density-classes in $\Sigma_{\le 2}$.

A key geometric fact is that diffeomorphism-invariant local scalar densities built from the metric are constructed from curvature tensors and their covariant derivatives, rather than from raw coordinate derivatives $\partial g$. In particular, curvature starts at second derivative order, so there are no genuinely invariant first-derivative scalar densities built only from $g$ and $\partial g$.

Within the truncation $\Sigma_{\le 2}$ this yields:
\begin{enumerate}
\item \textbf{Order 0.}
The only scalar density available is the volume form:
\[
[\sqrt{-g}]\ \in\ \Sigma_{\le 2}.
\]
Thus the associated action is the cosmological term
\[
S_\Lambda[g] := \int_M d^4x\, \sqrt{-g}.
\]
\item \textbf{Order 1.}
Formally, $\partial g$ is not a tensor, so no coordinate-invariant scalar density can be formed from $g$ and $\partial g$ alone. Intuitively, one may always choose locally inertial coordinates at a point, so that $\partial_\lambda g_{\mu\nu}=0$ there. 
\item \textbf{Order 2.}
At second derivative order the curvature tensor appears. The only non-trivial scalar density, up to normalisation and equivalence of terms with extra total derivatives, is generated by the Ricci scalar $R$:
\[
[\sqrt{-g}\,R]\ \in\ \Sigma_{\le 2},
\qquad
S_{\mathrm{EH}}[g]:=\int_M d^4x\, \sqrt{-g}\, R.
\]
Any other curvature scalar quadratic in curvature (e.g.\ $R^2$, $R_{\mu\nu}R^{\mu\nu}$) lies beyond the $\le 2$-derivative truncation because it contains four derivatives of the metric in total.
\end{enumerate}
Therefore, within $\Sigma_{\le 2}$,
\[
\operatorname{Inv}(\operatorname{Diff}(M))
=
\mathrm{span}\Big\{\, [\sqrt{-g}],\ [\sqrt{-g}\,R]\,\Big\}.
\]
\paragraph{Total package.}
Define the GR package
\[
P_{\mathrm{GR}}
:=
\mathrm{span}\Big\{\, [\sqrt{-g}],\ [\sqrt{-g}\,R]\,\Big\}
\ \subseteq\ \Sigma_{\le 2}.
\]
Equivalently, at the level of actions one may parameterise elements of $P_{\mathrm{GR}}$ by coefficients, as
\[
S[g]
=
\int_M d^4x\, \sqrt{-g}\,\big(\alpha R - 2\beta\big),
\]
with $\alpha,\beta$ constants to be identified empirically with $1/(16\pi G)$ and $\Lambda/(8\pi G)$ after normalisation.
\paragraph{Fixed point property.}
By construction, $P_{\mathrm{GR}}$ is stable:
\[
\F_{\mathrm{GR}}(P_{\mathrm{GR}})=P_{\mathrm{GR}},
\]
because $\operatorname{Sym}(P_{\mathrm{GR}})=\operatorname{Diff}(M)$ and $\operatorname{Inv}(\operatorname{Diff}(M))\cap\Sigma_{\le 2}=P_{\mathrm{GR}}$.
\paragraph{Least fixed pointhood.}
Since $\F_{\mathrm{GR}}$ is monotone on a complete lattice of packages (by the Galois correspondence), Tarski's theorem guarantees a least fixed point. In this example we compute directly that $\F_{\mathrm{GR}}(P_0)=P_{\mathrm{GR}}$, and $P_{\mathrm{GR}}$ is a fixed point, hence it is the least fixed point above $P_0$.

Hence, relative to the metric kinematics, diffeomorphism invariance, locality, and the two-derivative truncation, the Einstein--Hilbert plus cosmological term is the minimal stable residue. Therefore, GR is captured as the minimal stable package within the chosen domain, $\Sigma_{\le 2}$.

\section{Conclusion}

We began with a simple thought that the laws of nature might be characterised as exactly those candidates that satisfy an admissibility condition whose evaluation itself depends on the totality of laws. Taken extensionally, this idea either becomes ill-typed (laws versus law-sets) or collapses into self-membership once levels are forced to coincide. The first part of the paper made this failure explicit: An extensional reconstruction cannot sustain the intended explanatory reading without either type collapse or trivial self-endorsement.

The repair is to keep types distinct and treat "theory packages" as elements of a structured space of packages. In that setting, self-subsumption is not a paradoxical set-theoretic construction but a fixed point equation 
\[
P=\F(P).
\]
Under monotonicity on a complete lattice, Tarski's theorem guarantees fixed points and supplies a canonical minimal choice, the least fixed point $P=\mu\F$. This does not remove circularity; it replaces it by a controlled closure condition, where $P$ is determined as the minimal package stable under an explicitly specified admissibility operator.

To connect the abstract schema to physics, we proposed reading admissibility as completion under constitutive constraints. The key structural lesson is that closure is the invariant content. A package counts as an admissible theory precisely when it is stable under the completion map. This aligns naturally with an invariance-style conception of objectivity. What is "the theory" is what survives (is closed under) the admissibility transformations one treats as constitutive. In the examples, we instantiated $\F$ via an invariance--symmetry correspondence, and showed how QED and GR is represented as least fixed points once one fixes a candidate domain and encodes symmetry, locality in the gauge sense, and a truncation criterion.

Several limitations are built into the present formulation. First, the framework does not tell us which candidate domain $\Sigma$ or which admissibility operator $\F$ is correct; those encode substantive physical and philosophical choices. Second, the examples rely on simplifying truncations (renormalisable operators in QED; a two-derivative expansion in GR), so they illustrate the fixed point architecture rather than capture the full content of the theories. Third, monotonicity is a strong hypothesis. Many attractive criteria (consistency checks, optimality constraints, some notions of "minimality") are not monotone, and fall outside the direct scope of the Tarski canonicity argument.

These constraints also suggest natural directions for future work:
\begin{enumerate}[label=(\roman*)]
\item Enrich the package lattices $\mathcal T$ to incorporate
quotients by physical equivalence (gauge, duality, field redefinitions) as part of the package structure.
\item Study non-monotone admissibility notions and alternative fixed point principles, clarifying when and how canonicity survives without monotonicity.
\item Extend the invariance-based construction to more realistic settings (effective field theory bases, anomaly constraints, renormalisation group closure, and bootstrap-style constraint packages), where the set of admissible completions may form a nontrivial lattice of fixed points.
\end{enumerate}

The main takeaway is structural: Whenever one characterises a "fundamental" package by a rule of membership that itself depends on that package, the resulting formulation is, in general, a fixed point problem. Making this explicit separates two questions that are often conflated. The logical question of how to represent self-subsumption without contradiction, and the substantive question of which admissibility transformations should count as constitutive of physical law.

\printbibliography
\end{document}